\documentclass[12pt]{article}

\usepackage[margin=1in]{geometry}
\usepackage{CJKutf8, amsmath, amssymb, amsthm, authblk, hyperref}
\usepackage[T1]{fontenc}

\newtheorem{lemma}{Lemma}
\newtheorem{theorem}{Theorem}

\theoremstyle{definition}
\newtheorem{definition}{Definition}

\newtheorem{question}{Question}

\DeclareMathOperator{\cov}{cov}
\DeclareMathOperator{\tr}{tr}
\DeclareMathOperator*{\var}{Var}
\DeclareMathOperator*{\E}{\mathbb E}

\usepackage[style=trad-unsrt, citestyle=numeric-comp, giveninits=true, doi=false, url=false, eprint=false, isbn=false]{biblatex}

\begin{document}

\title{Deviations from maximal entanglement for eigenstates of the Sachdev-Ye-Kitaev model}

\author{Yichen Huang (黄溢辰)\thanks{yichenhuang@fas.harvard.edu}}
\author{Yi Tan\thanks{yitan@g.harvard.edu}}
\author{Norman Y. Yao\thanks{nyao@fas.harvard.edu}}
\affil{Department of Physics, Harvard University, Cambridge, Massachusetts 02138, USA}

\begin{CJK}{UTF8}{gbsn}

\maketitle

\end{CJK}

\begin{abstract}

We consider mid-spectrum eigenstates of the Sachdev-Ye-Kiteav (SYK) model. We prove that for subsystems whose size is a constant fraction of the  system size, the entanglement entropy  deviates from the maximum entropy by at least a positive constant. This result highlights the difference between the entanglement entropy of mid-spectrum eigenstates of the SYK model and that of random states.

\end{abstract}

\section{Introduction}

Characterizing the emergence of chaos in many-particle quantum systems is of fundamental interest to a diverse array of fields. In the context of statistical physics, it is relevant to our understanding of  thermalization \cite{DKPR16}, and in the context of high energy physics, to our understanding of the scrambling dynamics of black holes \cite{SS08}. Meanwhile, quantum chaos has also emerged as an essential concept within quantum computing, underlying the  demonstration of quantum advantage via random quantum circuits \cite{AAB+19}. Despite these myriad  connections, the technical toolsets available to investigate quantum chaos remain limited. For quantum many-body systems, brute-force numerical methods (i.e.~exact diagonalization) are usually limited to relatively small system sizes because their run time is exponential in the system size. At the same time, generic quantum chaotic systems are not analytically tractable.

One strategy for gaining analytical insight into quantum chaos is via random matrix theory (RMT). Such a strategy involves making certain assumptions. For example, one might assume that the Hamiltonian of a quantum chaotic system behaves like a random matrix from the Gaussian unitary ensemble (GUE). This assumption implies that the eigenstates are similar to Haar-random states. Another assumption could be that, at late times, the time-evolution operator behaves like a Haar-random unitary; this would imply that at sufficiently long times, the system is described by a Haar-random state. By using RMT to develop our analytical understanding of quantum chaotic systems, one must rely on such assumptions, which, while heuristic, are mathematically elegant. Indeed, such RMT descriptions are independent of the microscopic details of the system and thus may represent universal aspects of quantum chaotic behavior. As a testament to the success and applicability of RMT, in many specific models, predictions from random matrix theory have been numerically observed \cite{KH13, YCHM15, VR17, GG18, Hua19NPB, Hua21NPB, HMK22, KSVR23, RJK24} and even rigorously proved up to certain approximations~\cite{YGC23, HH23}.

One obvious place where physical systems differ from random matrix theory is in the structure of their interactions. For example, in physical quantum spin systems  (not necessarily on a lattice), interactions among spins are usually few-body in the sense that each term in the Hamiltonian acts non-trivially only on at most a constant number of spins. Since such few-bodyness of interactions is \emph{not} captured by RMT, it is important to understand to what extent the properties of quantum chaotic systems with few-body interactions deviate from the predictions of RMT.

In this paper, the property we choose to investigate is the entanglement of mid-spectrum eigenstates. ``Mid-spectrum'' means that the energy of the eigenstate is close to the mean energy (average of all eigenvalues) of the Hamiltonian.\footnote{For a traceless Hamiltonian, mid-spectrum eigenstates are those whose energy is close to zero.} We only consider mid-spectrum eigenstates because the energy of a random state is exponentially close to the mean energy of the Hamiltonian.

\begin{question}
In quantum many-body systems governed by chaotic Hamiltonians with few-body interactions, does the difference between the entanglement of mid-spectrum eigenstates and that of random states vanish in the thermodynamic limit?
\end{question}

We answer this question for the Sachdev-Ye-Kitaev (SYK) model \cite{SY93, Kit15, MS16}. The SYK model consists of Majorana fermions with random all-to-all four-body interactions and is maximally chaotic in the sense of being a fast scrambler \cite{Kit15, MSS16, MS16}. The entanglement of eigenstates within the SYK model was studied in Refs.~\cite{FS16, LCB18, HG19, ZLC20, Zha22}.

Recall that an eigenstate of a GUE Hamiltonian is a Haar-random state and that by definition, a GUE Hamiltonian exhibits no structure. In comparison, an SYK Hamiltonian has almost no structure, but exhibits fermion parity and few-bodyness. If there is a difference between the entanglement entropy of mid-spectrum eigenstates of the SYK model and that of random states, one expects this to be attributed to few-bodyness, since fermion parity does not significantly affect the entanglement:\footnote{Without any essentially new ideas, it is straightforward to extend our main result (Theorem \ref{t:m}) to a spin-glass model \cite{ES14} with random all-to-all two-body interactions. The model is an analogue of the SYK model in spin systems and does not conserve parity.} In the thermodynamic limit, random states with and without definite parity have the same average entanglement entropy. We begin by informally stating our main result: 

\begin{theorem} [informal statement] \label{t:m}
Let $A$ be a subsystem whose size is a constant fraction of the system size. In the thermodynamic limit, the entanglement entropy between $A$ and $\bar A$ of mid-spectrum eigenstates of the SYK model deviates from maximal entanglement by at least a positive constant.
\end{theorem}

The rest of our paper is organized as follows. Section \ref{sec:pre} introduces basic definitions and briefly reviews the entanglement of random states. Section \ref{sec:res} presents the main result, i.e.~Theorem \ref{t:main}, which is the formal version of Theorem \ref{t:m}. Section \ref{sec:skt} gives an outline of the proof, while Section \ref{app} provides technical details. We note a forthcoming companion manuscript, which  contains corroborating  numerical results and emphasizes the physical intuition underlying the technical theorems presented here. 

\section{Preliminaries} \label{sec:pre}

We will use  standard asymptotic notation. Let $f,g:\mathbb R^+\to\mathbb R^+$ be two functions. One writes $f(x)=O(g(x))$ if and only if there exist constants $M,x_0>0$ such that $f(x)\le Mg(x)$ for all $x>x_0$; $f(x)=\Omega(g(x))$ if and only if there exist constants $M,x_0>0$ such that $f(x)\ge Mg(x)$ for all $x>x_0$; $f(x)=\Theta(g(x))$ if and only if there exist constants $M_1,M_2,x_0>0$ such that $M_1g(x)\le f(x)\le M_2g(x)$ for all $x>x_0$; $f(x)=o(g(x))$ if and only if for any constant $M>0$ there exists a constant $x_0>0$ such that $f(x)<Mg(x)$ for all $x>x_0$.

Let $A$ be a subsystem and $\bar A$ be the complement of $A$ (rest of the system). Let $d_A,d_{\bar A}$ be the Hilbert space dimensions of subsystems $A,\bar A$, respectively. Assume without loss of generality that $d_A\le d_{\bar A}$.

\begin{definition} [entanglement entropy]
The entanglement entropy of a pure state $|\psi\rangle$ is defined as the von Neumann entropy
\begin{equation}
    S(\psi_A):=-\tr(\psi_A\ln\psi_A)
\end{equation}
of the reduced density matrix $\psi_A:=\tr_{\bar A}|\psi\rangle\langle\psi|$.
\end{definition}

Note that $\max_{|\psi\rangle}S(\psi_A)=\ln d_A$. Thus, $\ln d_A$ is referred to as maximal entanglement.

\begin{theorem} [conjectured and partially proved by Page \cite{Pag93}; proved in Refs.~\cite{FK94, San95, Sen96}] \label{t:page}
For a pure state $|\psi\rangle$ chosen uniformly at random with respect to the Haar measure,
\begin{equation} \label{eq:page}
\E_{|\psi\rangle}S(\psi_A)=\sum_{k=d_{\bar A}+1}^{d_Ad_{\bar A}}\frac1k-\frac{d_A-1}{2d_{\bar A}}=\ln d_A-\frac{d_A}{2d_{\bar A}}+\frac{O(1)}{d_Ad_{\bar A}}.
\end{equation}
\end{theorem}

Note that the second step of Eq.~(\ref{eq:page}) uses the formula
\begin{equation}
\sum_{k=1}^n\frac1k=\ln n+\gamma+\frac1{2n}+O(1/n^2)
\end{equation}
for $n=d_{\bar A}$ and $n=d_Ad_{\bar A}$, where $\gamma\approx0.577216$ is the Euler-Mascheroni constant.

The distribution of $S(\psi_A)$ is highly concentrated around the mean $\E_{|\psi\rangle}S(\psi_A)$ \cite{HLW06}. This can also be seen from the exact formula \cite{VPO16, Wei17} for the variance $\var_{|\psi\rangle}S(\psi_A)$.

In a system of $N$ qubits, let $A$ be a subsystem of size $L$. Suppose $f:=L/N$ is a fixed constant such that $0<f\le1/2$. Theorem \ref{t:page} implies that in the limit $N\to\infty$, 
\begin{equation}
    \E_{|\psi\rangle}S(\psi_A)=L\ln2-2^{(2f-1)N-1}+O(2^{-N}).
\end{equation}
Thus, for $0<f<1/2$, the difference between the entanglement entropy of random states and maximal entanglement is exponentially small $e^{-\Omega(N)}$. For an equal bipartition ($f=1/2$), the difference is exponentially close to $1/2$.

\section{Results} \label{sec:res}

Consider a system of $N$ Majorana fermions $\chi_1,\chi_2,\ldots,\chi_N$ with the anticommutation relation $\{\chi_j,\chi_k\}=2\delta_{jk}$, where $N$ is an even number and $\delta$ is the Kronecker delta. Let $[N]:=\{1,2,\ldots,N\}$ be the set of integers from $1$ to $N$ and
\begin{equation}
\binom{[N]}4:=\{J\subseteq[N]:|J|=4\}    
\end{equation}
be the set of all size-$4$ subsets of $[N]$.

\begin{definition} [Sachdev-Ye-Kitaev model \cite{SY93, Kit15, MS16}]
Let $K:=\{K_J\}_{J\in\binom{[N]}4}$ be a collection of $\binom{N}4$ independent real Gaussian random variables with zero mean $\overline{K_J}=0$ and unit variance $\overline{K_J^2}=1$. For any $J=\{j,k,l,m\}$ with $j<k<l<m$, define $X_J=\chi_j\chi_k\chi_l\chi_m$. The Hamiltonian of the SYK model is
\begin{equation} \label{eq:syk}
H_K=\frac1{\sqrt{\binom{N}4}}\sum_{J\in\binom{[N]}4}K_JX_J.
\end{equation}
\end{definition}

A complex version \cite{Sac15, GKST20} of the SYK model is also known as the embedded Gaussian unitary ensemble \cite{Kot01, BW03}. Historically, the ensemble was introduced in order to study the effect of few-bodyness on properties other than entanglement.

Let $A\subseteq[N]$ with $|A|=L\le N/2$, where $L\ge8$ is an even integer. $A$ can be understood as a subsystem (of size $L$) consisting of the Majorana fermions whose indices are in $A$. Let
\begin{equation}
    \binom{A}4:=\{J\subseteq A:|J|=4\}
\end{equation}
be the set of all size-$4$ subsets of $A$. Let $\E_{|A|=L}$ denote averaging over all subsystems of size $L$. There are $\binom{N}L$ such subsystems.

For each $K$, let $|\psi_K\rangle$ be an eigenstate of $H_K$. Let $\max_{|\psi_K\rangle}$ denote maximizing over all eigenstates of $H_K$. Let $\psi_{K,A}$ be the reduced density matrix of $|\psi_K\rangle$ on subsystem $A$.

\begin{theorem} \label{t:main}
For any $L\ge8$,
\begin{equation} \label{eq:main}    \Pr_K\left(\max_{|\psi_K\rangle}\E_{|A|=L}S(\psi_{K,A})=\frac{L\ln2}2-\Omega(L^8/N^8)\right)=1-\frac{O(1)}{L^2\max\{L^2,N\}}.
\end{equation}
\end{theorem}

The eigenstate $|\psi_K\rangle$ that maximizes $\E_{|A|=L}S(\psi_{K,A})$ is called the maximally entangled eigenstate of $H_K$. It is usually but not always a mid-spectrum eigenstate. Suppose $f:=L/N$ is a fixed constant such that $0<f\le1/2$. Theorem \ref{t:main} says that in the limit $N\to\infty$, with high probability the entanglement entropy of maximally entangled
eigenstates of the SYK model deviates from maximal entanglement by at least a positive constant.

Comparing Theorems \ref{t:page} and \ref{t:main}, the entanglement entropy of maximally entangled eigenstates of the SYK model is provably different from that of random states. The difference is $\Omega(1)$ if $N/2-C'>L=\Omega(N)$ for a certain constant $C'>0$. We conjecture that the difference is also $\Omega(1)$ if $N/2-C'\le L\le N/2$.

\section{Proof sketch} \label{sec:skt}

In this section, we sketch the proof of Theorem \ref{t:main} for $L=\Omega(N)$. For simplicity, we assume average behavior is typical behavior. The full proof for any $L\ge8$ with rigorous probabilistic analysis is given in Section \ref{app}.

Let $|\psi_K\rangle$ be an eigenstate of $H_K$ with eigenvalue $\lambda$. We have either $|\lambda|\ge1/2$ or $|\lambda|<1/2$.

\paragraph{Proof sketch for $|\lambda|\ge1/2$.}Since $\tr H_K=0$, the condition $|\lambda|\ge1/2$ means that the energy of $|\psi_K\rangle$ is significantly different from the mean energy of $H_K$. We use this observation to prove that the entanglement entropy of $|\psi_K\rangle$ deviates from maximal entanglement by $\Omega(1)$.

Let
\begin{equation}
H_{K|A}=\frac1{\sqrt{\binom{N}4}}\sum_{J\in\binom{A}4}K_JX_J.
\end{equation}
be the restriction of $H_K$ to subsystem $A$. Let
\begin{equation}
\sigma_{K|A}(\beta):=e^{-\beta H_{K|A}}/\tr(e^{-\beta H_{K|A}}),\quad E_{K|A}(\beta):=\tr(\sigma_{K|A}(\beta)H_{K|A})    
\end{equation}
be the thermal state of $H_{K|A}$ and its energy at inverse temperature $\beta$ so that $E_{K|A}(0)=0$. Since $H_{K|A}$ is, up to a prefactor, the SYK model in a system of $L$ Majorana fermions, the heat capacity of the model $\sqrt LH_{K|A}$ is extensive in the sense that it is  proportional to $L$. Thus,
\begin{equation} \label{eq:ebeta}
E_{K|A}(\beta)=-\Theta(\beta)
\end{equation}
for $|\beta|=o(\sqrt L)$.

Assume without loss of generality that $\lambda<-1/2$. The Hamiltonian $H_K$ (\ref{eq:syk}) has $\binom{N}4$ terms, $\binom{L}4$ of which are in $H_{K|A}$. In an average sense,
\begin{equation} \label{eq:eng}
    \tr(\psi_{K,A}H_{K|A})=\frac{\binom{L}4}{\binom{N}4}\lambda=-\Omega(1).
\end{equation}
Since the thermal state maximizes the entropy among all states with the same energy \cite{Weh78},
\begin{equation}
    S(\psi_{K,A})\le S(\sigma_{K|A}(\beta_*)),
\end{equation}
where $\beta_*$ is determined by
\begin{equation}
    E_{K|A}(\beta_*)=\tr(\psi_{K,A}H_{K|A}).
\end{equation}
Combining this equation with Eqs.~(\ref{eq:ebeta}), (\ref{eq:eng}), we find that $\beta_*=\Omega(1)$ is positive. Finally,
\begin{multline}
S(\sigma_{K|A}(\beta_*))=S(\sigma_{K|A}(0))+\int_0^{\beta_*}\,\mathrm dS(\sigma_{K|A}(\beta))=\frac{L\ln2}2+\int_0^{\beta_*}\,\mathrm\beta\,\mathrm dE_{K|A}(\beta)\\
=\frac{L\ln2}2-\int_0^{\beta_*}\Theta(\beta)\,\mathrm d\beta=\frac{L\ln2}2-\Theta(\beta_*^2)=\frac{L\ln2}2-\Omega(1).
\end{multline}

\paragraph{Proof sketch for $|\lambda|<1/2$.}

We may assume $\lambda=0$ so that $|\psi_K\rangle$ is exactly at the mean energy of $H_K$. The proof for other values of $|\lambda|<1/2$ is (almost) identical. Most terms in
\begin{equation}
H_K^2=\frac1{\binom{N}4}\sum_{J_1,J_2\in\binom{[N]}4}K_{J_1}K_{J_2}X_{J_1}X_{J_2}
\end{equation}
are traceless. Let
\begin{equation} \label{eq:defG}
G_K=\frac1{\binom{N}4}\sum_{J_1\neq J_2;J_1,J_2\in\binom{[N]}4}K_{J_1}K_{J_2}X_{J_1}X_{J_2}
\end{equation}
be the traceless part of $H_K^2$. We find that $|\psi_K\rangle$ is an eigenstate of $G_K$ with eigenvalue
\begin{equation}
    -\frac{\tr(H_K^2)}{2^{N/2}}=-\frac1{\binom{N}4}\sum_{J\in\binom{[N]}4}K_J^2=-1\pm O(1/N^2).
\end{equation}
Thus, the energy of $|\psi_K\rangle$ with respect to $G_K$ is significantly different from the mean energy of $G_K$. Since $G_K$ is a model with random all-to-all few-body interactions, the above proof for $|\lambda|>1/2$ can be adapted to $G_K$.

$G_K$ differs from $H_K$ in that its coefficients are random variables not independent from each other. The dependence leads to technical difficulties, which fortunately can be handled rigorously.

\section{Proof of Theorem \ref{t:main}} \label{app}

Since $\tr(X_{J_1}X_{J_2})=0$ for $J_1\neq J_2$,
\begin{equation}
    \frac{\tr(H_K^2)}{2^{N/2}}=\frac1{\binom{N}4}\sum_{J\in\binom{[N]}4}K_J^2
\end{equation}
is distributed as $\frac1{\binom{N}4}\chi_{\binom{N}4}^2$, where $\chi_k^2$ denotes the chi-square distribution with $k$
degrees of freedom. It follows directly from the tail bound \cite{LM00} for the chi-square distribution that
\begin{lemma}
\begin{equation}
\Pr_K\left(\frac{\tr(H_K^2)}{2^{N/2}}\ge1/2\right)=1-e^{-\Omega(N^4)}.
\end{equation}
\end{lemma}

Let
\begin{equation}
C_k:=\prod_{j=0}^{k-1}\frac{N-j}{L-j}
\end{equation}
for $k=4,5,\ldots,8$ and
\begin{equation}
    C:=\frac{C_8}{C_4}=\frac{(N-4)(N-5)(N-6)(N-7)}{(L-4)(L-5)(L-6)(L-7)}=\Theta(N^4/L^4).
\end{equation}

Let
\begin{equation} \label{eq:HKA}
H_{K,A}=\frac{C_4}{\sqrt{\binom{N}4}}\sum_{J\in\binom{A}4}K_JX_J.
\end{equation}
Since $H_{K,A}/C_4$ is the restriction of $H_K$ to subsystem $A$,
\begin{equation}
    \E_{|A|=L}H_{K,A}\otimes I_{\bar A}=H_K,
\end{equation}
where $I_{\bar A}$ is the identity operator on subsystem $\bar A$.

Let $\E_K$ denote averaging over all random variables $K_J|_{J\in\binom{[N]}{4}}$.

\begin{lemma} [\cite{FTW19}] \label{l:H1}
For any positive even integer $n$,
\begin{equation}
\E_K\frac{\tr(H_{K,A}^n)}{2^{L/2}}\le C_4^{n/2}(n-1)!!.
\end{equation}
\end{lemma}

\begin{proof}
We include the proof of this lemma for completeness. From the definition (\ref{eq:HKA}) of $H_{K,A}$,
\begin{equation}
\E_K\frac{\tr(H_{K,A}^n)}{2^{L/2}}=\frac{C_4^n}{\binom{N}{4}^{n/2}}\sum_{J_1,J_2,\ldots,J_n\in\binom{A}4}\E_K(K_{J_1}K_{J_2}\cdots K_{J_n})\frac{\tr(X_{J_1}X_{J_2}\cdots X_{J_n})}{2^{L/2}}.
\end{equation}
We observe that
\begin{gather}
    \E_K(K_{J_1}K_{J_2}\cdots K_{J_n})\ge0,\quad\forall J_1,J_2,\ldots,J_n,\\
    \left|\frac{\tr(X_{J_1}X_{J_2}\cdots X_{J_n})}{2^{L/2}}\right|\le1.
\end{gather}
Therefore,
\begin{multline}
\E_K\frac{\tr(H_{K,A}^n)}{2^{L/2}}\le\frac{C_4^n}{\binom{N}{4}^{n/2}}\sum_{J_1,J_2,\ldots,J_n\in\binom{A}4}\E_K(K_{J_1}K_{J_2}\cdots K_{J_n})=\frac{C_4^n}{\binom{N}{4}^{n/2}}\E_K\left(\left(\sum_{J\in\binom{A}4}K_J\right)^n\right)\\
=\frac{C_4^n}{\binom{N}{4}^{n/2}}\binom L4^{n/2}(n-1)!!=C_4^{n/2}(n-1)!!,
\end{multline}
where we used the fact that $\sum_{J\in\binom{A}4}K_J$ is a real Gaussian random variable with zero mean and variance $\binom{L}4$.
\end{proof}

\begin{lemma} [\cite{FTW21}] \label{l:H2}
For any integer $n\ge2$,
    \begin{equation}
        \var_K\frac{\tr(H_{K,A}^n)}{2^{L/2}}\le\frac{C_4^n}{\binom{L}4}2^nn!n^2.
    \end{equation}
\end{lemma}

\begin{proof}
We include the proof of this lemma for completeness. For a tuple $(J_1,J_2,\ldots,J_{2n})\in\binom{A}{4}^{2n}$ (of even length), let
\begin{equation}
\# J:=|\{1\le j\le 2n:J_j=J\}|
\end{equation}
be the number of occurrences of a particular $J\in\binom{A}4$ in the tuple. Let
\begin{equation}
    P_n:=\{(J_1,J_2,\ldots,J_{2n}):\#J_1=\#J_2=\cdots=\#J_{2n}=2\}\subseteq\binom{A}4^{2n}
\end{equation}
so that 
\begin{equation}
    |P_n|=(2n-1)!!\prod_{j=0}^{n-1}\left(\binom{L}4-j\right).
\end{equation}
Let
\begin{equation}
    \Delta_k:=\{J_j:1\le j\le n\}\cap\{J_j:n+1\le j\le2n\}
\end{equation}
be the set of common elements of the first and second halves of the tuple. Let $I$ be the identity operator. $P_n$ can be decomposed as $P_n=P_{n,0}\cup P_{n,1}\cup P_{n,2}$, where
\begin{gather}
P_{n,0}=\{(J_1,J_2,\ldots,J_{2n})\in P_n:\Delta_k=\emptyset\},\\
P_{n,1}=\{(J_1,J_2,\ldots,J_{2n})\in P_n:\Delta_k\neq\emptyset,X_{J_1}X_{J_1}\cdots X_{J_n}=\pm I\},\\
P_{n,2}=\{(J_1,J_2,\ldots,J_{2n})\in P_n:X_{J_1}X_{J_1}\cdots X_{J_n}\neq\pm I\}
\end{gather}
are pairwise disjoint subsets.

We now estimate $|P_{n,1}|$. Let $m:=|\Delta_k|$. $(J_1,J_2,\ldots,J_{2n})\in P_n$ and $X_{J_1}X_{J_1}\cdots X_{J_n}=\pm I$ imply that the product of all elements of the set $\Delta_k$ is $\pm I$. Due to this constraint, for fixed $m$, the number of choices of $\Delta_k$ is upper bounded by $\binom{L}{4}^{m-1}/m!$. For fixed $\Delta_k$, the number of choices of $(\{J_j:1\le j\le n\},\{J_j:n+1\le j\le 2n\})$ is $\binom{\binom{L}{4}-m}{n-m}\binom{n-m}{\frac{n-m}2}$. For fixed $\{J_j:1\le j\le n\}$ and $\{J_j:n+1\le j\le 2n\}$, the number of choices of $(J_1,J_2,\ldots,J_{2n})$ is $(n!)^2/2^{n-m}$. Thus,
\begin{multline}
|P_{n,1}|\le\sum_{\substack{m:~0<m\le n\\n-m~\textnormal{even}}}\binom{\binom{L}{4}-m}{n-m}\binom{n-m}{\frac{n-m}2}\frac{(n!)^2\binom{L}{4}^{m-1}}{2^{n-m}m!}\le\sum_{\substack{m:~0<m\le n\\n-m~\textnormal{even}}}\frac{\binom{L}{4}^{n-m}}{(n-m)!}\frac{(n!)^2\binom{L}{4}^{m-1}}{m!}\\
=n!\binom{L}{4}^{n-1}\sum_{\substack{m:~0<m\le n\\n-m~\textnormal{even}}}\binom{n}{m}\le2^nn!\binom{L}4^{n-1}.
\end{multline}

From the definition (\ref{eq:HKA}) of $H_{K,A}$,
\begin{equation}
    \var_K\frac{\tr(H_{K,A}^n)}{2^{L/2}}=\frac{C_4^{2n}}{2^L\binom{N}4^n}\left(\sum_{(J_1,J_2,\ldots,J_{2n})\in P_n}+\sum_{(J_1,J_2,\ldots,J_{2n})\in \binom{A}4^{2n}\setminus P_n}\right)\cdots=:V_1^H+V_2^H,
\end{equation}
where the summand, denoted by ``$\cdots$,'' is
\begin{equation}
    \cov(K_{J_1}K_{J_2}\cdots K_{J_{n}},K_{J_{n+1}}K_{J_{n+2}}\cdots K_{J_{2n}})\tr(X_{J_1}X_{J_2}\cdots X_{J_n})\tr(X_{J_{n+1}}X_{J_{n+2}}\cdots X_{J_{2n}}).
\end{equation}
If $(J_1,J_2,\ldots,J_{2n})\in P_{n,0}$, then $K_{J_1}K_{J_2}\cdots K_{J_{n}}$ and $K_{J_{n+1}}K_{J_{n+2}}\cdots K_{J_{2n}}$ are independent from each other so that $\cov(K_{J_1}K_{J_2}\cdots K_{J_{n}},K_{J_{n+1}}K_{J_{n+2}}\cdots K_{J_{2n}})=0$. If $(J_1,J_2,\ldots,J_{2n})\in P_{n,2}$, then $\tr(K_{J_1}K_{J_2}\cdots K_{J_n})=0$. Hence,
\begin{equation} \label{eq:V1H}
    |V_1^H|\le\frac{C_4^{2n}}{\binom{N}4^n}\sum_{(J_1,J_2,\ldots,J_{2n})\in P_{n,1}}\E_K(K_{J_1}K_{J_2}\cdots K_{J_{2n}})=\frac{C_4^{2n}|P_{n,1}|}{\binom{N}4^n}\le2^nn!\frac{C_4^n}{\binom{L}4}.
\end{equation}
Moreover,
\begin{align} \label{eq:V2H}
    |V_2^H|&\le\frac{C_4^{2n}}{\binom{N}4^n}\sum_{(J_1,J_2,\ldots,J_{2n})\in\binom{A}4^{2n}\setminus P_n}
    |\cov(K_{J_1}K_{J_2}\cdots K_{J_{n}},K_{J_{n+1}}K_{J_{n+2}}\cdots K_{J_{2n}})|\nonumber\\
    &\le\frac{C_4^{2n}}{\binom{N}4^n}\sum_{(J_1,J_2,\ldots,J_{2n})\in\binom{A}4^{2n}\setminus P_n}
    \E(K_{J_1}K_{J_2}\cdots  K_{J_{2n}})\nonumber\\
    &=\frac{C_4^{2n}}{\binom{N}4^n}\left(\sum_{(J_1,J_2,\ldots,J_{2n})\in\binom{A}4^{2n}}-\sum_{(J_1,J_2,\ldots,J_{2n})\in P_n}\right)
    \E(K_{J_1}K_{J_2}\cdots  K_{J_{2n}})\nonumber\\
    &=\frac{C_4^{2n}}{\binom{N}4^n}\left(\E\left(\left(\sum_{J\in\binom{A}4}K_J\right)^{2n}\right)-|P_n|\right)=C_4^n(2n-1)!!\left(1-\prod_{j=0}^{n-1}\frac{\binom{L}4-j}{\binom{L}4}\right)\nonumber\\
    &\le C_4^n(2n-1)!!\sum_{j=0}^{n-1}\frac{j}{\binom{L}4}=C_4^n(2n-1)!!\frac{n(n-1)}{2\binom{L}4}.
\end{align}
Therefore,
\begin{equation}
    \var_K\frac{\tr(H_{K,A}^n)}{2^{L/2}}\le|V_1^H|+|V_2^H|\le\frac{C_4^n}{\binom{L}4}\left(2^nn!+(2n-1)!!\frac{n(n-1)}2\right)\le\frac{C_4^n}{\binom{L}4}2^nn!n^2.
\end{equation}
\end{proof}

Recall the definition (\ref{eq:defG}) of $G_K$. Let
\begin{equation} \label{eq:GKA}
G_{K,A}=\frac1{\binom{N}4}\sum_{J_1,J_2\in\binom{A}4:J_1\neq J_2}C_{|J_1\cup J_2|}K_{J_1}K_{J_2}X_{J_1}X_{J_2}.
\end{equation}
Note that $5\le|J_1\cup J_2|\le8$ for $J_1\neq J_2$. $C_5<C_6<C_7<C_8$ are chosen such that
\begin{equation} \label{eq:ave}
    \E_{|A|=L}G_{K,A}\otimes I_{\bar A}=G_K.
\end{equation}

\begin{lemma} \label{l:G1}
For any positive even integer $n$,
\begin{equation}
\E_K\frac{\tr(G_{K,A}^n)}{2^{L/2}}\le C^n(2n-1)!!.
\end{equation}
\end{lemma}

\begin{proof}
The proof of Lemma \ref{l:G1} is similar to that of Lemma \ref{l:H1}. From the definition (\ref{eq:GKA}) of $G_{K,A}$,
\begin{multline}
    \E_K\frac{\tr(G_{K,A}^n)}{2^{L/2}}=\frac1{{\binom{N}4}^n}\sum_{\substack{J_1,J_2,\ldots,J_{2n}\in\binom{A}4\\J_1\neq J_2;J_3\neq J_4;\ldots;J_{2n-1}\neq J_{2n}}}C_{|J_1\cup J_2|}C_{|J_3\cup J_4|}\cdots C_{|J_{2n-1}\cup J_{2n}|}\\
    \times\E_K(K_{J_1}K_{J_2}\cdots K_{J_{2n}})\frac{\tr(X_{J_1}X_{J_2}\cdots X_{J_{2n}})}{2^{L/2}}.
\end{multline}
Using
\begin{equation}
    C_{|J_1\cup J_2|}C_{|J_3\cup J_4|}\cdots C_{|J_{2n-1}\cup J_{2n}|}\le C_8^n,
\end{equation}
we obtain
\begin{align} \label{eq:moment}
&\E_K\frac{\tr(G_{K,A}^n)}{2^{L/2}}\le\frac{C_8^n}{{\binom{N}4}^n}\sum_{\substack{J_1,J_2,\ldots,J_{2n}\in\binom{A}4\\J_1\neq J_2;J_3\neq J_4;\ldots;J_{2n-1}\neq J_{2n}}}\E_K(K_{J_1}K_{J_2}\cdots K_{J_{2n}})\nonumber\\
&\le\frac{C_8^n}{{\binom{N}4}^n}\sum_{J_1,J_2,\ldots,J_{2n}\in\binom{A}4}\E_K(K_{J_1}K_{J_2}\cdots K_{J_{2n}})=\frac{C_8^n}{{\binom{N}4}^n}\E_K\left(\left(\sum_{J\in\binom{A}4}K_J\right)^{2n}\right)\nonumber\\
&=\frac{C_8^n}{{\binom{N}4}^n}{\binom{L}4}^n(2n-1)!!=C^n(2n-1)!!.
\end{align}
\end{proof}

\begin{lemma} \label{l:G2}
For any integer $n\ge2$,
    \begin{equation}
        \var_K\frac{\tr(G_{K,A}^n)}{2^{L/2}}\le\frac{C^{2n}}{\binom{L}4}2^{2n}(2n)!n^2.
    \end{equation}
\end{lemma}

\begin{proof}
The proof of Lemma \ref{l:G2} is similar to that of Lemma \ref{l:H2}. From the definition (\ref{eq:GKA}) of $G_{K,A}$,
\begin{equation}
    \var_K\frac{\tr(G_{K,A}^n)}{2^{L/2}}=\frac{1}{2^L\binom{N}4^{2n}}\left(\sum_{\substack{(J_1,J_2,\ldots,J_{4n})\in P_{2n}\\J_1\neq J_2;J_3\neq J_4;\ldots;J_{4n-1}\neq J_{4n}\\}}+\sum_{\substack{(J_1,J_2,\ldots,J_{4n})\in\binom{A}4^{4n}\setminus P_{2n}\\J_1\neq J_2;J_3\neq J_4;\ldots;J_{4n-1}\neq J_{4n}\\}}\right)\cdots=:V_1^G+V_2^G,
\end{equation}
where the summand, denoted by ``$\cdots$,'' is
\begin{multline}
C_{|J_1\cup J_2|}C_{|J_3\cup J_4|}\cdots C_{|J_{4n-1}\cup J_{4n}|}\cov(K_{J_1}K_{J_2}\cdots K_{J_{2n}},K_{J_{2n+1}}K_{J_{2n+2}}\cdots K_{J_{4n}})\\
    \times\tr(X_{J_1}X_{J_2}\cdots X_{J_{2n}})\tr(X_{J_{2n+1}}X_{J_{2n+2}}\cdots X_{J_{4n}}).
\end{multline}
Similar to (\ref{eq:V1H}),
\begin{equation}
    |V_1^G|\le\frac{1}{\binom{N}4^{2n}}\sum_{(J_1,J_2,\ldots,J_{4n})\in P_{2n,1}}C_8^{2n}=\frac{C_8^{2n}|P_{2n,1}|}{\binom{N}4^{2n}}\le2^{2n}(2n)!\frac{C^{2n}}{\binom{L}4}.
\end{equation}
Similar to (\ref{eq:V2H}),
\begin{align}
    |V_2^G|&\le\frac{1}{\binom{N}4^{2n}}\sum_{(J_1,J_2,\ldots,J_{4n})\in\binom{A}4^{4n}\setminus P_{2n}}C_8^{2n}
    |\cov(K_{J_1}K_{J_2}\cdots K_{J_{2n}},K_{J_{2n+1}}K_{J_{2n+2}}\cdots K_{J_{4n}})|\nonumber\\
    &\le\frac{C_8^{2n}}{\binom{N}4^{2n}}\sum_{(J_1,J_2,\ldots,J_{4n})\in\binom{A}4^{4n}\setminus P_{2n}}
    \E(K_{J_1}K_{J_2}\cdots  K_{J_{4n}})\nonumber\\
    &=\frac{C_8^{2n}}{\binom{N}4^{2n}}\left(\sum_{(J_1,J_2,\ldots,J_{4n})\in\binom{A}4^{4n}}-\sum_{(J_1,J_2,\ldots,J_{4n})\in P_{2n}}\right)
    \E(K_{J_1}K_{J_2}\cdots  K_{J_{4n}})\nonumber\\
    &=\frac{C_8^{2n}}{\binom{N}4^{2n}}\left(\E\left(\left(\sum_{J\in\binom{A}4}K_J\right)^{4n}\right)-|P_{2n}|\right)=C^{2n}(4n-1)!!\left(1-\prod_{j=0}^{2n-1}\frac{\binom{L}4-j}{\binom{L}4}\right)\nonumber\\
    &\le C^{2n}(4n-1)!!\sum_{j=0}^{2n-1}\frac{j}{\binom{L}4}=C^{2n}(4n-1)!!\frac{n(2n-1)}{\binom{L}4}.
\end{align}
Therefore,
\begin{equation}
    \var_K\frac{\tr(G_{K,A}^n)}{2^{L/2}}\le|V_1^G|+|V_2^G|\le\frac{C^{2n}}{\binom{L}4}\big(2^{2n}(2n)!+(4n-1)!!n(2n-1)\big)\le\frac{C^{2n}}{\binom{L}4}2^{2n}(2n)!n^2.
\end{equation}
\end{proof}

\begin{lemma}
For any integer $n\ge2$,
\begin{gather}
\Pr_K\left(\left|\E_{|A|=L}\frac{\tr(H_{K,A}^n)}{2^{L/2}}\right|\le (2C_4)^{n/2}n^2\sqrt{n!}\right)\ge1-\frac{O(1)}{n^2L^2\max\{L^2,N\}},\\
\Pr_K\left(\left|\E_{|A|=L}\frac{\tr(G_{K,A}^n)}{2^{L/2}}\right|\le (2C)^nn^2\sqrt{(2n)!}\right)\ge1-\frac{O(1)}{n^2L^2\max\{L^2,N\}}.
\end{gather}
\end{lemma}

\begin{proof}
Using Lemma \ref{l:H2},
\begin{align}
    &\var_K\E_{|A|=L}\frac{\tr(H_{K,A}^n)}{2^{L/2}}=\frac1{\binom{N}L^2}\sum_{|A|=L,|A'|=L}\cov\left(\frac{\tr(H_{K,A}^n)}{2^{L/2}},\frac{\tr(H_{K,A'}^n)}{2^{L/2}}\right)\nonumber\\
    &\le\frac1{\binom{N}L^2}\sum_{|A|=L,|A'|=L,A\cap A'\neq\emptyset}\frac12\left(\var\frac{\tr(H_{K,A}^n)}{2^{L/2}}+\var\frac{\tr(H_{K,A'}^n)}{2^{L/2}}\right)\nonumber\\
    &=\E_{|A|=L}\frac{|\{A':|A'|=L,A\cap A'\neq\emptyset\}|}{\binom{N}L}\var\frac{\tr(H_{K,A}^n)}{2^{L/2}}=\left(1-\frac{\binom{N-L}L}{\binom{N}L}\right)\E_{|A|=L}\var_K\frac{\tr(H_{K,A}^n)}{2^{L/2}}\nonumber\\
    &\le\min\{1,2L^2/N\}\frac{C_4^n}{\binom{L}4}2^nn!n^2.
\end{align}
Similarly, using Lemma \ref{l:G2},
\begin{equation}    \var_K\E_{|A|=L}\frac{\tr(G_{K,A}^n)}{2^{L/2}}\le\min\{1,2L^2/N\}\frac{C^{2n}}{\binom{L}4}2^{2n}(2n)!n^2.
\end{equation}
We complete the proof using Chebyshev inequality.
\end{proof}

To prove Theorem \ref{t:main}, it suffices to prove
\begin{equation} \label{eq:2prove}
\E_{|A|=L}S(\psi_{K,A})=\frac{L\ln2}2-\Omega(L^8/N^8)
\end{equation}
for any eigenstate $|\psi_K\rangle$ of $H_K$ under the assumptions that
\begin{gather}
\frac{\tr(H_K^2)}{2^{N/2}}\ge1/2,\label{eq:H2}\\
\left|\E_{|A|=L}\frac{\tr(H_{K,A}^n)}{2^{L/2}}\right|\le (2C_4)^{n/2}n^2\sqrt{n!},\quad\forall n\ge2,\label{eq:aH}\\
\left|\E_{|A|=L}\frac{\tr(G_{K,A}^n)}{2^{L/2}}\right|\le (2C)^nn^2\sqrt{(2n)!},\quad\forall n\ge2.\label{eq:aG}
\end{gather}

Let
\begin{equation}
\sigma^H_{K,A}(\beta):=e^{-\beta H_{K,A}}/\tr(e^{-\beta H_{K,A}}),\quad\sigma_{K,A}^G(\beta):=e^{-\beta G_{K,A}}/\tr(e^{-\beta G_{K,A}})    
\end{equation}
be the thermal states of $H_{K,A}$ and $G_{K,A}$, respectively, at inverse temperature $\beta$. Let
\begin{equation}
\mathcal E_K^H(\beta):=\E_{|A|=L}\tr(\sigma_{K,A}^H(\beta)H_{K,A}),\quad\mathcal E_K^G(\beta):=\E_{|A|=L}\tr(\sigma_{K,A}^G(\beta)G_{K,A}).
\end{equation}
so that $\mathcal E_K^H(0)=\mathcal E_K^G(0)=0$. As long as $K$ has at least two nonzero entries, both $\mathcal E_K^H$ and $\mathcal E_K^G$ are strictly monotonically decreasing.

\begin{lemma} \label{l:energy}
Let $c=\Theta(1)$ be a sufficiently small positive constant. For $0\le\beta\le c/\sqrt{C_4}$,
\begin{equation}
    \mathcal E_K^H(\beta)\ge-12\beta C_4-O(\beta^2C_4^{3/2}).
\end{equation}
For $0\le\beta\le c/C$,
\begin{equation}
    \mathcal E_K^G(\beta)\ge-80\beta C^2-O(\beta^2C^3).
\end{equation}
\end{lemma}

\begin{proof}
Since $H_{K,A}$ and $G_{K,A}$ are traceless,
\begin{equation} \label{eq:par}
    \tr(e^{-\beta H_{K,A}})\ge2^{L/2},\quad\tr(e^{-\beta G_{K,A}})\ge2^{L/2},\quad\forall A.
\end{equation}
Using (\ref{eq:par}) and (\ref{eq:aH}),
\begin{multline}
    \mathcal E_K^H(\beta)=\E_{|A|=L}\tr(\sigma_{K,A}^H(\beta)H_{K,A})\ge\E_{|A|=L}\frac{\tr(e^{-\beta H_{K,A}}H_{K,A})}{2^{L/2}}=\sum_{n=0}^\infty\E_{|A|=L}\frac{(-\beta)^n\tr(H_{K,A}^{n+1})}{n!2^{L/2}}\\
    \ge-\sum_{n=2}^\infty\frac{\beta^{n-1}(2C_4)^{n/2}n^2\sqrt{n!}}{(n-1)!}.
\end{multline}
For $0\le\beta\le c/\sqrt{C_4}$, the last series above is convergent. Similarly,
\begin{multline}
    \mathcal E_K^G(\beta)=\E_{|A|=L}\tr(\sigma_{K,A}^G(\beta)G_{K,A})\ge\E_{|A|=L}\frac{\tr(e^{-\beta G_{K,A}}G_{K,A})}{2^{L/2}}=\sum_{n=0}^\infty\E_{|A|=L}\frac{(-\beta)^n\tr(G_{K,A}^{n+1})}{n!2^{L/2}}\\
    \ge-\sum_{n=2}^\infty\frac{\beta^{n-1}(2C)^nn^2\sqrt{(2n)!}}{(n-1)!}.
\end{multline}
For $0\le\beta\le c/C$, the last series above is convergent.
\end{proof}

Let
\begin{equation}
    \mathcal S_K^H(\beta):=\E_{|A|=L}S(\sigma_{K,A}^H(\beta)),\quad \mathcal S_K^G(\beta):=\E_{|A|=L}S(\sigma_{K,A}^G(\beta))
\end{equation}
so that $\mathcal S_K^H(0)=\mathcal S_K^G(0)=L(\ln2)/2$. As long as $K$ has at least two nonzero entries, both $\mathcal S_K^H$ and $\mathcal S_K^G$ are strictly monotonically decreasing (increasing) for positive (negative) $\beta$.

\begin{lemma} \label{l:thermo}
For $\beta$ such that $0\le-\mathcal E_K^H(\beta)=O(\sqrt{C_4})$,
\begin{equation} \label{eq:entH}
    \mathcal S_K^H(\beta)=\frac{L\ln2}2-\Omega(\mathcal E_K^H(\beta))^2/C_4.
\end{equation}
For $\beta$ such that $0\le-\mathcal E_K^G(\beta)=O(C)$,
\begin{equation} \label{eq:entG}
    \mathcal S_K^G(\beta)=\frac{L\ln2}2-\Omega(\mathcal E_K^G(\beta))^2/C^2.
\end{equation}
\end{lemma}

\begin{proof}
Lemma \ref{l:energy} implies that
\begin{equation}
\beta=\Omega(-\mathcal E_K^H(\beta)/C_4).
\end{equation}
Combining this with the thermodynamic relation
\begin{equation}
    \mathrm d\mathcal S_K^H(\beta)/\mathrm d\beta=\beta\,\mathrm d\mathcal E_K^H(\beta)/\mathrm d\beta\implies\mathrm d\mathcal S_K^H(\beta)/\mathrm d\mathcal E_K^H(\beta)=\beta,
\end{equation}
we obtain Eq. (\ref{eq:entH}). Equation (\ref{eq:entG}) can be proved similarly.
\end{proof}

Let $|\psi_K\rangle$ be an eigenstate of $H_K$ with eigenvalue $\lambda$.

\begin{lemma} \label{l:mainH}
If $|\lambda|=\Omega(1)$, then
\begin{equation} \label{eq:mainH}
    \E_{|A|=L}S(\psi_{K,A})=\frac{L\ln2}2-\Omega(L^4/N^4).
\end{equation}
\end{lemma}

\begin{proof}
Since
\begin{equation}
    \E_{|A|=L}\tr(\psi_{K,A}H_{K,A})=\langle\psi_K|H_K|\psi_K\rangle=\lambda,
\end{equation}
an upper bound on $\E_{|A|=L}S(\psi_{K,A})$ can be obtained as follows. For each $A$, we introduce a density matrix $\rho_A$ supported on $A$. We maximize $\E_{|A|=L}S(\rho_A)$ subject to the constraint
\begin{equation} \label{eq:ec}
    \E_{|A|=L}\tr(\rho_AH_{K,A})=\lambda.
\end{equation}
Lemma 11 in Ref.~\cite{Hua21ISIT} or Ref.~\cite{Hua22TIT} implies that the maximum is achieved when $\rho_A=\sigma_{K,A}^H(\beta)$, where the inverse temperature $\beta$ is determined from
\begin{equation}
    \mathcal E_K^H(\beta)=\lambda.
\end{equation}

Assume without loss of generality that $\lambda<0$. Lemma \ref{l:thermo} implies that
\begin{equation}
    \E_{|A|=L}S(\psi_{K,A})\le\mathcal S_K^H(\beta)=\frac{L\ln2}2-\Omega(\lambda^2/C_4)=\frac{L\ln2}2-\Omega(L^4/N^4).
\end{equation}
\end{proof}

\begin{lemma} \label{l:mainG}
If $|\lambda|\le1/2$, then
\begin{equation} \label{eq:mainG}
    \E_{|A|=L}S(\psi_{K,A})=\frac{L\ln2}2-\Omega(L^8/N^8).
\end{equation}
\end{lemma}

\begin{proof}
Using (\ref{eq:ave}) and (\ref{eq:H2}),
\begin{equation}
\E_{|A|=L}\tr(\psi_{K,A}G_{K,A})=\langle\psi_K|G_K|\psi_K\rangle=\langle\psi_K|H_K^2|\psi_K\rangle-\tr(H_K^2)/2^{N/2}\le\lambda^2-1/2\le-1/4.
\end{equation}
As in the proof of Lemma \ref{l:mainH},
\begin{equation}
    \E_{|A|=L}S(\psi_{K,A})\le\mathcal S_K^G(\beta),
\end{equation}
where $\beta$ satisfies
\begin{equation}
    \mathcal E_K^G(\beta)\le-1/4.
\end{equation}
Lemma \ref{l:thermo} implies that
\begin{equation}
    \mathcal S_K^G(\beta)=\frac{L\ln2}2-\Omega(1/C^2)=\frac{L\ln2}2-\Omega(L^8/N^8).
\end{equation}
\end{proof}

Equation (\ref{eq:2prove}) follows from Lemmas \ref{l:mainH} and \ref{l:mainG}.

\paragraph{Note added.}Recently, we became aware of related work which explores constraints on the entanglement entropy originating from energy variance in spatially local models~\cite{ghosh2024late}.

\section*{Acknowledgments}

We would like to thank Gregory D. Kahanamoku-Meyer for discussions and collaboration on related work, Nicole Yunger Halpern for a conversation that motivated this work, and Aram W. Harrow for helpful discussions. This work was supported by the U.S. Department of Energy through the Quantum Information Science Enabled Discovery (QuantISED) for High Energy Physics (KA2401032) and through the GeoFlow Grant No.~de-sc0019380. N.Y.Y. acknowledges support from a Simons Investigatorship.

\printbibliography

\end{document}